\documentclass[12 pt]{article}
\usepackage[centertags]{amsmath}
\usepackage{amsfonts}
\usepackage{amssymb}
\usepackage{amsthm}
\usepackage{graphicx}
 \usepackage{color}
\usepackage{dsfont}

\newcommand{\q}{\mathbb{Q}}
\newcommand{\rr}{\mathbb{R}}
\newcommand{\n}{\mathbb{N}}
\newcommand{\h}{\mathbb{H}}

\newcommand{\bi}{\mathbf{i}}

\newcommand{\bj}{\mathbf{j}}

\newcommand{\ket}[1]{\vert#1\rangle}

\newcommand{\B}{\mathbb B}

\theoremstyle{plain}
\newtheorem{thm}{Theorem}[section]

\theoremstyle{definition}

\newtheorem{defn}{Definition}[section]
\theoremstyle{remark}

\numberwithin{equation}{section}

\usepackage{color}

\begin{document}
\title{Algorithmic complexity of quantum capacity}

\author{Samad Khabbazi Oskouei \\
\textit{\small Department of Mathematics, Islamic Azad University}\\
\textit{\small Varamin-Pishva Branch, 33817-7489, Iran}\\
\and Stefano Mancini \\
\textit{\small School of Science and Technology,
University of Camerino}\\
\textit{\small Via M. delle Carceri 9, I-62032 Camerino, Italy}\\
\textit{\small \& INFN--Sezione Perugia}\\
\textit{\small Via A. Pascoli, I-06123 Perugia, Italy}}

\maketitle
\begin{abstract}
Recently the theory of communication developed by Shannon has been extended to the quantum realm by exploiting the rules of quantum theory. This latter stems on complex vector spaces. However complex (as well as real) numbers are just idealizations and they are not available in practice where we can only deal with rational numbers. This fact naturally leads to the question of whether the developed notions of capacities for quantum channels truly catch their ability to transmit information. Here we  address this issue for the quantum capacity. To this end we resort to the notion of semi-computability in order to approximately (by rational numbers) describe quantum states and quantum channel maps. We introduce algorithmic entropies (like algorithmic quantum coherent information) and derive relevant properties for them.

Then, we show that quantum capacity based on semi-computable concept equals
the entropy rate of algorithmic coherent information which in turn equals the standard quantum capacity. We also prove that such a quantity is recursively approximable. Although the algorithmic coherent information results not computable, at the end we present a method to compute it on a restricted subset of density matrices.

\end{abstract}

\textit{Keywords:} Algorithmic complexity, quantum entropies, quantum channels, quantum capacity

\section{Introduction}

Quantum channels are maps on the set of quantum states (density operators)
that are linear, completely positive and trace preserving \cite{wilde}.
They generalize the notion of classical channels, that is stochastic maps acting on probability distributions.
As such they permit to transfer quantum information (besides classical one),
namely quantum correlations also known as entanglement.
Their quantum information transmission ability would be captured by a capacity notion (quantum capacity),
similarly to what happen in the classical information theoretical framework
with the Shannon capacity of classical channels.
A remarkable difference is that the quantum capacity requires a regularization formula \cite{wilde}.
That is the computation of an entropic rate over infinite many quantum channel uses.
This is due to the possible effects arising by employing entangled input across different channel uses. Due to that, quantum capacity evaluation results a daunting task.

One way to deal with this problem can be to write down a computer program in classical computer or even in a quantum one and try to evaluate the quantity expressed by capacity formula. In this way, operators produced by algorithms are acceptable by a computer. However, the set of density matrices on a Hilbert spaces of any dimension is uncountable, while the set that can be acceptable by a Turing machine (even a quantum one) is countable.
This motivate us to study the quantum capacity when restricting to the usage of operators (and super-operators) that can be accepted by a computer.

To this end we invoke the notion of semi-computability introduced in \cite{Zvonkin70thecomplexity}  for a given function $f:\mathbb{N}\to \mathbb{R}$ and then extended to density matrices and operators in separable Hilbert spaces \cite{Gacs, FSA, SKHO}.
The main features of semi-computability is the universality concept, that is the existence of a
semi-measure
$\mu$ which dominates all other semi-computable semi-measures up to a positive multiplicative constant \cite{Zvonkin70thecomplexity}.
A universal semi-measure is related to the concept of algorithmic complexity in that
$-\log\mu(i_1\ldots i_n)\leq K(i_1\ldots i_n)+c$ and $K(i_1\ldots i_n)\leq -\log\mu(i_1\ldots i_n)+ c'$ where $i_1\ldots i_n$ is a sequence of symbols from a finite alphabet,
$c, c' > 0$ are constant independent of $i_1\ldots i_n$, and $K$ is the Kolmogorov complexity.
The latter notion was developed by \cite{Solomonoff}, Kolmogorov \cite{Kolmogorov} and Chaitin \cite{Chaitin}.
In a nutshell, the complexity of a target object is measured by the difficulty to describe it; in the case of targets describable by binary strings, they are algorithmically complex when their shortest binary descriptions are essentially of the same length in terms of necessary bits, the descriptions being binary programs such that any universal Turing machine that runs them outputs the target string.
Taking this approach, the quantum capacity should be characterized in terms of algorithmic (Kolmogorov) complexity. There are several ways this complexity can be extended to the quantum realm \cite{Berthiaume, Vitanyi, Mora, Gacs}.  Here we follow the Gacs approach \cite{Gacs}. This is  based on the notion of universal semi-density matrix. That is, on the existence of a density matrix $\hat{\mu}$ on separable Hilbert space that dominates any other semi-computable density matrix up to a multiple positive constant number.  The  algorithmic complexity for a semi-computable density matrix $\rho$ is given by $-{\rm Tr}(\rho \log\hat{\mu})$.

Along this way, we algorithmically define mutual and coherent information
in quantum systems.
 The nice feature is that these quantities are linear with respect to their arguments, a property that does not hold true for the standard entropies. Nonetheless we will prove that their rates equal those of the standard entropies.
 Then we show that quantum capacity based on semi-computable concept, namely
the entropy rate of algorithmic coherent information, equals the standard quantum capacity. We also prove that such a quantity is recursively approximable, which is one step ahead in investigating the computability or non-computability of quantum channel capacity \cite{wolf,cubitt}. Although the algorithmic coherent information results not computable, at the end we present a method to compute it on a restricted subset of density matrices which might be useful for computer programmers.

The organization of the paper foresees an initial Section \ref{sec:pre} where we recall some basic notions and set the notation. Then Section \ref{sec:aqi} revisits relevant entropic quantities from an algorithmic point of view and contains the derivation of some of their relevant properties,
 including the fact that the algorithmic quantum capacity coincides with the standard one.
In Section \ref{sec:aqc} one can find the proof that this is a recursively approximable quantity. Furthermore an algorithmic method to compute it on a restricted subset of density matrices is presented there.
Section \ref{sec:conclu} is for conclusions and outlook.


\section{Preliminaries}\label{sec:pre}

A quantum channel is a completely positive and trace preserving linear map $\Phi:\B(\h_A)\to \B(\h_B)$, where $\B(\h)$ is the algebra of bounded linear operators defined on the Hilbert space $\h$.
Actually input states to a quantum channel are bounded operators of unit trace (density operators)
constituting a proper subset $\mathbb{S}(\h)\subset\B(\h)$.

Given a quantum channel $\Phi$ and a density operator $\rho$, there are three important entropic quantities related to the pair $(\rho,\Phi)$, namely the entropy of the input state\footnote{Throughout the paper the $\log$ is intended on base 2.}
\begin{equation*}
S(\rho):=-{\rm Tr}\left(\rho\log\rho\right),
\end{equation*}
the output entropy
\begin{equation*}
S(\Phi(\rho)):=-{\rm Tr}\left(\Phi(\rho)\log\Phi(\rho)\right),
\end{equation*}
and the exchange entropy
\begin{equation*}
S(\rho,\Phi):=S\left({\rm id}\otimes \Phi)(|\psi_\rho\rangle_{RA}\langle\psi_\rho|\right),
\end{equation*}
where  $|\psi_\rho\rangle_{RA}\in\h_R\otimes\h_A$ is the purification of $\rho$.
Let the isometry $V:\B(\h_A)\to\B(\h_B)\otimes\B(\h_E)$ be the Stinespring dilation of the channel $\Phi$ \cite{Stine},
where $E$ labels the environment, i.e. $\Phi(\rho)={\rm Tr}_E(V \rho V^\dag)$.
Defining by
\begin{equation}
\tilde{\Phi}(\rho):={\rm Tr}_B(V \rho V^\dag)
\end{equation}
the complementary channel, it results \cite{wilde}:
$$
S(\rho, \Phi)=S\left(\tilde{\Phi}(\rho)\right).
$$
The input, output and exchange entropy are building blocks for defining other entropic quantities like
quantum coherent information
\begin{equation}
I_c(\rho,\Phi):=S(\Phi(\rho))-S(\rho,\Phi).
\end{equation}
The corresponding entropy rate is
\begin{eqnarray}
Q_c(\Phi)&:=&\lim_{n\to\infty}\frac{1}{n} \max_{\rho_n} I_c(\rho_n, \Phi^{\otimes n}),
\label{eq:Qcapacity2}
\end{eqnarray}
where $\rho_n\in\mathbb{S}(\h_A^{\otimes n})$.  In Refs.\cite{devetak,hayden} it is shown that the quantum capacity for a quantum channel $\Phi$ is given by $Q_c(\Phi)$.

Furthermore, for any two density operators
$\rho,\sigma\in\mathbb{S}(\h)$, their relative entropy is defined by
\begin{equation}\label{eq:relent}
S(\rho,\sigma):={\rm Tr}\left(\rho (\log\rho-\log\sigma)\right).
\end{equation}

It is known \cite{lindblad} that
\begin{equation}\label{relative}
S\left(\Phi(\rho),\Phi(\sigma) \right) \leq S (\rho,\sigma).
\end{equation}

 The aim of this paper is to revisit the characterization \eqref{eq:Qcapacity2}
 of a quantum channel $\Phi$ by an algorithmic approach.

Let us first recall few basic notion about computability and algorithmic complexity.

\bigskip

A function from $\n^k$ to $\n$ is called partially computable if it
is computed by a Turing Machine \cite{Da} (here
Turing Machine, program, and algorithm are used interchangeably). A partial computable function is called computable if it is halted  by the associated Turing Machine for every input natural number.

Let $\theta:\mathbb{N}\times \mathbb{N}\to \mathbb{N}$ be a paring function defined by $\theta(x, y)=2^x(2 y +1)-1$, for each $x, y\in \mathbb{N}$. Let $\iota$ be a bijection from $\mathbb{N}\times\mathbb{N}\to \mathbb{Q}$ ~\cite{FSA}. A function $f:\mathbb{N}\to \mathbb{Q}$ is called computable if there exists a computable function $g:\mathbb{N}\to\mathbb{N}$ such that $f(n)=\iota\circ\theta^{-1}(g(n))$.

A function
$g:\n\to\rr$ is called \textit{lower semi-computable} if there
exists a computable function $f:\n\times\n\to\q$,    such that the
sequence $f_n(x)=f(n,x)$, for any $x, n\in\mathbb{N}$, is increasing and
$\lim_{n\rightarrow\infty}f_n=g$.
A function $h:\n\to\rr$ is called \textit{upper semi-computable} if
$-h$ is lower semi-computable and it will be called
\textit{computable} if it is lower and upper semi-computable.

A real number $x\in\rr$ is called computable if there exists a computable function $f:\n\to \n\times \n$ such that $|x-f(n)|<2^{-n}$. For example $e$ and $\pi$ are computable real numbers.
A real number $x\in\rr$ is called \emph{recursively approximable} if there exists a computable function $f:\n\to \mathbb{Q}$ such that $x=\lim_{n\to\infty} f(n)$.
This means that we can approximate $x$ by $\q$ since  $\q$ is isomorphic with $\n\times\n$.

A function $\mu:\n\to\rr$ is called a \textit{semi-computable,
semi-measure} if it is a positive semi-computable function such that
$\Sigma_x\mu(x)\leq 1$. It has been shown in  \cite{FSA} that there exists
 a universal semi-computable  semi-measure $\mu$ in the following sense:

\textit{For any semi-computable semi-measure $\nu$ there exists a constant number
$c_\nu>0$ such that for each $x\in\n$, $c_\nu\,\nu(x)\leq\mu(x)$.}


\bigskip


Turning to the quantum world, we will denote by $\h$ the infinite dimensional separable Hilbert space obtained by the closure of the union of the nested infinite sequence
$\h^{\otimes n}\subset\h^{\otimes (n+1)}$ with respect to the norm coinciding with the
usual Hilbert norm on each $\h^{\otimes n}$.\footnote{Following \cite{FSA}
the embedding is obtained by turning the last bit of each canonical basis element to $0$.}
The corresponding orthogonal projections from $\h$ onto $\h^{\otimes n}$
will be denoted by  $P_n$, and  the canonical injected subspace $\h^{\otimes n}$ into $\h$ will be identified  by  $\h^{\otimes n}$.

The concept of semi-computable semi-density matrices on infinite
dimensional separable Hilbert spaces are introduced in \cite{FSA}. In such a context
we consider a fixed orthonormal basis $\{\ket{\bi}\}_{\bi\in\Omega^*_2}$,  where  $\Omega^*_2$ is the set of all finite length binary strings. Indeed, for any $\bi, \bj\in\Omega^*_2$, $\langle\bi | \bj\rangle=1$, if $\bi=\bj$ and $0$ otherwise.

A vector $\ket{\psi}=\sum_{\bi\in\Omega^*_2}a_\bi \ket{\bi} \in \mathbb{H}$ is
called \emph{elementary} if only a finite number of its coefficient
$a_\bi$ is not zero and those are algebraic numbers. It can be shown that each elementary
quantum state can be identified by a  natural number \cite{SKHO}.
Furthermore, $\ket{\psi}=\sum_{\bi\in\Omega^*_2}a_\bi \ket{\bi}\in\h$
where $a_i\in\rr$,
will be termed semi-computable if there exist a computable sequence
of elementary vectors $\ket{ \psi_n} = \sum_{i\in \n} a_{n,i} \ket{ \bi}$ and a
computable function $k: \n \rightarrow \q$, such that
$\lim_{n\rightarrow\infty} k_n = 0$, and for each $n$ it is
$0\leq a_\bi -a_{n,i}\leq k_n$.

A linear operator
$T:\h_{[0, n-1]}\rightarrow \h_{[0, n-1]}$, $\h_{[0, n-1]}:=\mathbb{H}^{\otimes n}$, will be called elementary if
the real and imaginary parts of all of its matrix entries are
algebraic numbers.
The linear operator $T:\h\rightarrow \h$, is a semi-density matrix
if  $T$ is positive and $0\leq {\rm Tr}(T)\leq 1$.

Let  $n_1,n_2\in\n$ and $n_1\leq n_2$. Let
$T_j:\h_{[0 ,n_j-1]}\rightarrow \h_{[0,n_j-1]}$, $j=1,2$, be two
linear operators: $T_2$ will be said to be quasi-greater than $T_1$,
written as $T_1\leq_q T_2$, if $P_{n_1}\,T_2\,P_{n_1}-T_1\geq0$, where
$P_{n_1}$ is the canonical projector from $\h_{n_2}$ to $\h_{n_1}$.
A sequence of linear operators $T_n:\h_{[0,n-1]}\rightarrow
\h_{[0,n-1]}$ will be called quasi-increasing if for all $n\geq 1$,
$T_{n+1}\geq_q T_n$.

A linear operator $T$ on $\h$ is a
semi-computable semi-density matrix, if there exists a computable
quasi-increasing sequence of elementary semi-density matrices
$T_n\in \B(\h_{[0,n-1]})\subseteq \B(\h)$ such that
$\lim_{n\rightarrow\infty}\|T-T_n\|_{1}=0$.

It has been shown in  \cite{FSA} that there exists  a universal
semi-computable semi-density matrix $\hat{\mu}$ in the following
sense:

\textit{For any semi-computable semi-density matrix $\rho$ there exists
a constant number $c_\rho>0$ such that $c_\rho\,\rho\leq\hat{\mu}$.}

\bigskip

\begin{defn}\label{def:Gacs}
The upper Gacs complexity for any semi-computable semi-density matrix
$\rho\in\mathbb{S}(\h)$ is defined as follows:
\begin{equation}
 G(\rho):=-{\rm Tr}(\rho \log\hat{\mu}).
\end{equation}
\end{defn}

 Let $AB$ be a composite quantum system with two subsystems $A$ and $B$. The associated Hilbert space $\mathbb{H}_{AB}$ can be extended to infinite dimensional separable Hilbert space $\mathbb{H}$ containing all  Hilbert subspaces $\mathbb{H}^{\otimes n}_{AB}$, $n\in\n$.

Then, $\hat{\mu}_{A B}:=P_{\dim \mathbb{H_{A B}}} \hat{\mu} P_{\dim \mathbb{H_{A B}}}$
(with $P_{\dim \mathbb{H_{A B}}}:\mathbb{H}\to\mathbb{H}_{A B}$ the canonical projector
 onto finite dimensional Hilbert space $\mathbb{H}_{AB}$),
$\hat{\mu}_A:={\rm Tr}_B(\hat{\mu}_{A B}) $ and $\hat{\mu}_B:={\rm Tr}_A(\hat{\mu}_{A B}) $, , are universal semi-density matrices on $\mathbb{H}_{A B}$, $\mathbb{H}_A$ and $\mathbb{H}_B$, respectively. In addition, we can establish  universal semi-density matrices  on Hilbert spaces $\mathbb{H}_{A B}^{\otimes n}$, $\mathbb{H}_A^{\otimes n}$ and $\mathbb{H}_B^{\otimes n}$, respectively, as follows: $$\hat{\mu}_{A B}^n:=P_{\dim \mathbb{H}_{A B}^{\otimes n}} \hat{\mu} P_{\dim \mathbb{H}_{A B}^{\otimes n}}, \quad \hat{\mu}_B^n:={\rm Tr}_A(\hat{\mu}^n_{A B}),\quad
\hat{\mu}_A^n:= {\rm Tr}_B(\hat{\mu}^n_{A B}),$$


\section{Algorithmic Quantum Information}\label{sec:aqi}

\begin{defn}
A quantum channel $\Phi:\B(\h_A)\to\B(\h_B)$ is defined semi-computable if for any semi-computable semi-density matrix $\rho\in\mathbb{S}(\h_A)$, $\Phi(\rho)$ is a semi-computable semi-density matrix on $\mathbb{S}(\h_B)$.
\end{defn}

 For example the quantum channels  with algebraic entries  in the Choi matrix representation \cite{Choi} are semi-computable channels.
\begin{thm}\label{relativequ}
Let $\h_A,\h_B$ be finite dimensional Hilbert spaces and $\Phi:\B(\h_A)\to \B(\h_B)$ be a semi-computable quantum channel. Then,
\begin{equation*}
S(\Phi^{\otimes n}(\rho_n),\hat{\mu}^n_B) \leq S(\rho_n, \hat{\mu}^n_A)+\alpha(n),
\end{equation*}
where  $\hat{\mu}^n_A$ and $\hat{\mu}^n_B$, for each $n$, are universal semi-density matrices on $\mathbb{H}^{\otimes n}_A$ and $\mathbb{H}^{\otimes n}_B$ respectively, and $\lim_{n\to \infty}\frac{\alpha(n)}{n}=0$.
\end{thm}

\begin{proof}
Since $\Phi$ is semi-computable then $\Phi^{\otimes n}(\hat{\mu}^n_A)$, for any $n$, is also semi-computable semi-density matrix. Now, let us consider the semi-computable semi-density matrix $\sum_{n=1}^{\infty} \delta(n) \Phi^{\otimes n}(\hat{\mu}_A^n)$ where
\begin{equation}
\delta(n)=\frac{1}{n\log^2{n}}.
\label{deltan}
\end{equation}
There exists a constant number $c_\Phi$ such that
$$c_\Phi \delta(n) \Phi^{\otimes n}(\hat{\mu}_A^n)\leq c_\Phi \sum_1^\infty \delta(n) \Phi^{\otimes n}(\hat{\mu}^n_A)\leq  \hat{\mu}.$$
Then,
\begin{eqnarray*}
c_\Phi \delta(n) \Phi^{\otimes n}(\hat{\mu}_A^n)&=&c_\Phi \delta(n) {\rm Tr}_A(P_{\dim\mathbb{H}_{AB}^{\otimes n}}\Phi^{\otimes n}(\hat{\mu}_A^n)P_{\dim\mathbb{H}_{AB}^{\otimes n}})\\
&\leq & {\rm Tr}_A(P_{\dim\mathbb{H}_{AB}^{\otimes n}}\hat{\mu}P_{\dim\mathbb{H}_{AB}^{\otimes n}})=\hat{\mu}_B^n.
\end{eqnarray*}
Furthermore,
$$-{\rm Tr}\left(\Phi^{\otimes n}(\rho_n) \log\hat{\mu}^n_B \right) \leq - {\rm Tr} \left(\Phi^{\otimes n}(\rho_n)\log\Phi^{\otimes n}(\hat{\mu}^n_A)\right)-\log c_\Phi-\log \delta(n).
$$
By adding ${\rm Tr}\left(\Phi^{\otimes n}(\rho_n) \log\Phi^{\otimes n}(\rho_n) \right)$ to both sides of inequality we obtain
$$
S(\Phi^{\otimes n}(\rho_n),\, \hat{\mu}^n_B)\leq S\left(\Phi^{\otimes n}(\rho_n),\, \Phi^{\otimes n}(\hat{\mu}^n_A)\right)-\log c_\Phi-\log \delta(n).
$$
By the relation \ref{relative} the proof is complete.
\end{proof}

\begin{defn}
The algorithmic coherent information for a given quantum channel
$\Phi:\B(\h_A)\to \B(\h_B)$ is defined as follows:
\begin{equation}
IG_c(\rho,\Phi):=G(\Phi(\rho))-G(\rho, \Phi),
\end{equation}
where following Definition \ref{def:Gacs}
$$
G(\rho):=-{\rm Tr}(\rho \log\hat{\mu}_A), \quad
G(\rho, \Phi):=-{\rm Tr}\left((id\otimes \Phi)|\psi\rangle_{RA}\langle\psi|_{RA} \log\hat{\mu}_{RB}\right),
$$
being $\ket{\psi}_{RA}$ a purification of $\rho\in\mathbb{S}(\h_A)$.
\end{defn}

\begin{defn}
For a given quantum channel $\Phi:\B(\h_A)\to \B(\h_B)$ the algorithmic coherent information entropy rate is defined as follows:
\begin{equation}
QG_c(\Phi):=\lim_{n\to\infty}\frac{1}{n} \max_{\rho_n} IG_c(\rho_n, \Phi^{\otimes n}),
\end{equation}
where $\rho_n\in \mathbb{S}(\mathbb{H}^{\otimes n}_A)$ are semi-computable density matrices.
\end{defn}

\begin{thm}\label{quantumalgorithmiccapacity}
For a semi-computable quantum channel  $\Phi:\B(\h_A)\to \B(\h_B)$, we have
$$
Q_c(\Phi)=QG_c(\Phi),
$$
where the maximum in the information entropy rate is taken over all semi-computable density matrices.
\end{thm}

\begin{proof}
First we show that
$$
QG_c(\Phi)\leq Q_c(\Phi).
$$
Let $\rho_n\in \mathbb{S}(\mathbb{H}_A^{\otimes n})$ be a computable sequence of
semi-computable semi-density matrices.
 Defining $\rho:=\sum_{n=1}^\infty \delta(n)\tilde{\Phi}^{\otimes n}(\rho_n)$, where
$\delta(n)$ is like in Eq.\ref{deltan}, it is clear that $\rho$ is  semi-computable
semi-density matrix and hence there exists a constant number $c_B>0$,
such that
$$
c_B \rho\leq \hat{\mu}\Rightarrow c_B \delta(n)\rho_n \leq \hat{\mu} \Rightarrow -\log\hat{\mu} \leq -\log c_B
-\log\delta(n)-\log\tilde{\Phi}^{\otimes n}(\rho_n).
$$
Now, let $P_{A B}:\mathbb{H}\to \mathbb{H}_{A B}^{\otimes n}$ be the canonical projector. Then, we have
\begin{eqnarray*}
-\log\hat{\mu}_B^n&=& -{\rm Tr}_A (P_{A B}\log\hat{\mu}P_{A B})\\
&\leq& - \log c_B -\log\delta(n)- {\rm Tr}_B( P_{AB}\log\tilde{\Phi}^{\otimes n}(\rho_n) P_{AB})\\
&=&-\log c_A -\log\delta(n)- \log\tilde{\Phi}^{\otimes n}(\rho_n).
\end{eqnarray*}
Therefore
\begin{eqnarray*}
IG_c(\rho_n,\Phi^{\otimes n})&=&G(\Phi^{\otimes n}(\rho_n))-G(\tilde{\Phi}^{\otimes n}(\rho_n)) \\
&\leq& ‎S(\Phi^{\otimes n}_E(\rho_n)) ‎- ‎\log\delta(n)\log c_B
-‎G(\tilde{\Phi}^{\otimes n}(\rho_n)) \\
&\leq& ‎S(\Phi^{\otimes n}_E(\rho_n)) ‎- ‎\log\delta(n)\log c_B
-‎S(\tilde{\Phi}^{\otimes n}(\rho_n)) \, .
\end{eqnarray*}
 The first inequality comes from the definition of Gacs complexity \ref{def:Gacs} and the fact that $\Phi^{\otimes n}(\rho_n)$ is semi-computable semi-density matrix.
For the second inequality, we used the fact that $S(\rho)\leq G(\rho)$.
By taking the limit on both sides of  the inequality we get
$$
\lim_{n\to \infty} \frac{1}{n} IG_c(\rho_n, \Phi^{\otimes n})\leq \lim_{n\to \infty} \frac{1}{n}
I_c(\rho_n, \Phi^{\otimes n})\leq  \lim_{n\to \infty} \frac{1}{n}
\max_\rho I_c(\rho, \Phi^{\otimes n}) \leq Q_c(\Phi).
$$
Therefore,
$$
QG_c(\Phi) \leq Q_c(\Phi).
$$

\bigskip

To prove the inverse relation, namely
$$
QG_c(\Phi)\geq Q_c(\Phi),
$$
we may notice from Theorem \ref{relativequ} that
$$
S(\Phi^{\otimes n}(\rho_n),\hat{\mu}_B)\leq S(\rho_n,\hat{\mu}_A)+\alpha(n),
$$
which using relative entropy (Eq.\eqref{eq:relent}) and Gacs complexity (Definition \ref{def:Gacs}) yields
$$
S(\Phi^{\otimes n}(\rho_n))-G(\Phi^{\otimes n}(\rho_n))
\leq S(\rho_n,\Phi^{\otimes n})-G(\rho_n,\Phi^{\otimes n})+\alpha(n).
$$
Thus rearranging l.h.s. and r.h.s. terms
$$
I_c(\rho_n, \Phi^{\otimes n})\leq IG_c(\rho_n, \Phi^{\otimes n})+\alpha(n).
$$
Finally, taking the limit on both sides of inequality, having in mind that
$\lim_{n\to\infty}\frac{\alpha(n)}{n}=0$, gives the desired result.
\end{proof}


\section{Algorithmic Quantum Capacity}\label{sec:aqc}

 Here we show that the standard quantum capacity can be approximated by a quantity linear in the channel's input. 

\begin{thm}\label{th:approx}
Let $\Phi:\B(\h_A)\to\B(\h_B)$ be a semi-computable quantum channel. Then its quantum capacity $Q_c(\Phi)$ is a recursively approximable number.
\end{thm}
\begin{proof}
Let $\rho_n\in\mathbb{S}(\mathbb{H}^{\otimes n})$ be a arbitrary density matrix with the following decomposition
$$
\rho_n=\sum_{i_1 \ldots i_n \,\, j_1 \ldots j_n} \lambda_{i_1 \ldots i_n \,\, j_1 \ldots j_n} |i_1 \ldots i_n><j_1 \ldots j_n| .
$$
Let also $V_n$ be the Stinespring dilation of $\Phi^{\otimes n}$ \cite{Stine}.
Since $\rho_n$ and $V_n$  are linear we have
\begin{eqnarray*}
IG_c(\rho_n) &=& -{\rm Tr}\left(\Phi^{\otimes n}(\rho_n) \log\hat{\mu}_{B^n}\right)
+{\rm Tr}\left(\tilde{\Phi}^{\otimes n}(\rho_n) \log\hat{\mu}_{E^n}\right) \\
&=&
\sum_{i_1 \ldots i_n \,\, j_1 \ldots j_n} \lambda_{i_1 \ldots i_n \,\, j_1 \ldots j_n} t_{i_1 \ldots i_n \,\, j_1 \ldots j_n},\\
\end{eqnarray*}
where we defined
\begin{eqnarray*}
t_{i_1 \ldots i_n \,\, j_1 \ldots j_n}&:=&-{\rm Tr}(\Phi^{\otimes n}(|i_1 \ldots i_n><j_1 \ldots j_n|) \log\hat{\mu}_{B^n})\\
&+&{\rm Tr}(\tilde{\Phi}^{\otimes n}( |i_1 \ldots i_n><j_1 \ldots j_n|) \log\hat{\mu}_{E^n}).
\end{eqnarray*}
These can be intended as entries of a matrix $T$.
If we now write $\rho$ and $T$ as vectors $v_\rho$ and $v_T$, respectively, then it follows
$$
|IG_c(\rho_n)|^2=|\langle v_\rho|v_T\rangle|^2\leq || |v_\rho\rangle|| \,\, || |v_T\rangle|| \leq || |v_T\rangle||.
$$
Next, let us consider $\lambda$ as the largest eigenvalue of $|T|$ with eigenvector $\ket{\Lambda}$ and define
$\rho=|\Lambda\rangle\langle\Lambda|$. It is clear that $||\rho||=1$ and hence
$$
\max_\rho |\langle v|w\rangle|=|| |v_T\rangle||=|\lambda|.
$$
Since $\hat{\mu}_{B^n}$ and $\hat{\mu}_{E^n}$ are semi-computable semi-density matrix, then there exist quasi-increasing computable sequences of elementary matrices that convergence to them. In turn it is know that the entries of elementary matrices are computable numbers, hence $t_{i_1 \ldots i_n \,\, j_1 \ldots j_n}$'s are recursively approximable numbers. This means that we can find a computable sequence of eigenvalues whose corresponding rate converges to the quantum capacity.
Therefore, the maximum of $IG_c(\rho_n)$ can be derived in an algorithmic way and hence, invoking the results of Theorem \ref{quantumalgorithmiccapacity}, we may conclude that the quantum channel capacity is approximated by this algorithmic method.
\end{proof}

\bigskip

 It is worth remarking that by means of Theorem \ref{th:approx} we remove the maximization of algorithmic coherent information at each level $n$ in proving the recursive approximability of $QG_c(\Phi)$.
 Furthermore $QG_c(\Phi)$, hence $Q_c(\Phi)$, results approximable by a quantity ($IG_c(\Phi)$) that is linear in the channel's input.

Although Theorem \ref{th:approx} sheds light on computational aspects of quantum capacity, it has the drawback of resorting to a universal semi-computable density matrix which is not computable.

 We henceforth show a case where the algorithmic coherent information can be exactly (with any desired degree of precision) computed and then quantum channel capacity is recursively approximated.

\bigskip

Quite generally, to find the quantum channel capacity by means of a computer program, we need a finite set of computable density matrices at each level $n$, the dimension of the Hilbert space.
We assume that the cardinality of such set obeys the condition $\lim_{n\to\infty}\log[f (n)/n]=0$,
with $f(n)$ a computable function denoting the number of chosen density matrices at level $n$. For example, computable polynomial functions have this property.
Then, the set can be represented as follows
\begin{eqnarray*}
{\cal S}&:=&\underbrace{\rho_{1},\ldots, \rho_{f(1)}}_{f(1)}, \;
\underbrace{\rho_{1+f(1)}, \rho_{2+f(1)},\ldots, \rho_{f(2)+f(1)}}_{f(2)}, \; \ldots \\
&&\ldots, \; \underbrace{\rho_{1+f(n-1)+\ldots+f(2)+f(1)},\ldots,  \rho_{f(n)+f(n-1)+\ldots+f(2)+f(1)}}_{f(n)}, \; \ldots
\end{eqnarray*}
Define
$$
\tilde {\rho}=\sum_{i=1} \delta(i) \rho_i,
$$
with $\delta(i)$ as in Eq.\ref{deltan}.
We emphasize that $\tilde {\rho}$ is not universal semi-computable likewise $\hat{\mu}$.

\begin {thm}\label{th:QonS}
Let $\Phi:\B(\h_A)\to\B(\h_B)$ be a computable quantum channel.
The quantum capacity $Q_c(\Phi)$ restricted to the set $\cal S$ results:
$$
Q_c \left(\Phi\right) = \lim_{n\to\infty}\frac{1}{n}\max_{\rho_{i_n}\in{\cal S}} \left\{-{\rm Tr}\left[
\Phi^{\otimes n}(\rho_{i_n})\log\Phi^{\otimes n}(\tilde{\rho})\right]
+{\rm Tr}\left[ \tilde {\Phi}^{\otimes n}(\rho_{i_n})\log\tilde {\Phi}^{\otimes n}(\tilde{\rho})\right]\right\}.
$$
\end {thm}

\begin {proof}
Let us assume that $\{\rho_{i_1}, \rho_{i_2},  \ldots  \}$ is a set of density matrices taken from
${\cal S}$.
By considering the position of each element in this set, one can define the following density matrix
$$
\sigma=\sum_{n=1} \delta (i_n) \rho_{i_n}.
$$

It is trivial to show that $\sigma \leq \tilde{\rho}$. Then, we have
\begin{equation}\label{semi1}
  -{\rm Tr}( \rho_{i_n}\log\tilde{\rho})\leq  -{\rm Tr}( \rho_{i_n}  \log\rho_{i_n})-\log  \delta (i_n)  .
\end{equation}
On the other hand, we have the following relation using the relative entropy
\begin{equation}\label{semi2}
   -{\rm Tr}( \rho_{i_n}  \log\rho_{i_n})\leq  -{\rm Tr}( \rho_{i_n}\log\tilde{\rho}).
\end{equation}

Now, it is
\begin{eqnarray*}
I_c (\rho_{i_n}, \Phi^{\otimes n})&\leq&   -{\rm Tr}\left[ \Phi^{\otimes n}(\rho_{i_n})\log\Phi^{\otimes n}(\tilde{\rho})\right]+{\rm Tr}\left[\tilde {\Phi}^{\otimes n}(\rho_{i_n})\log\tilde {\Phi}^{\otimes n}(\tilde{\rho})\right] \\
&&-\log  \delta(i_n).
\end{eqnarray*}
Therefore, the maximum of coherent information over all density matrices from the set $\cal S$ is less that the maximum of the r.h.s. over $\cal S$.  Taking into account that
$i_n\leq f(1)+f(2)+\dots+f(n)\leq nf(n)$, hence
$$
-\log\delta(i_n)\leq -\log\delta(n f(n))\Rightarrow
\lim_{n\to\infty}\frac {\log  \delta (i_n)}{n}=0,
$$
we get
$$
Q_c (\Phi) \leq  \lim\frac{1}{n}\max_{\rho_{i_n}} \left\{-{\rm Tr}\left[\Phi^{\otimes n}(\rho_{i_n})\log\Phi^{\otimes n}(\tilde{\rho})\right]+{\rm Tr}\left[   \tilde {\Phi}^{\otimes n}(\rho_{i_n})\log\tilde {\Phi}^{\otimes n}(\tilde{\rho})\right]\right\}.
$$
On the other hand, using the relations \ref{semi1} and \ref{semi2}, we have
\begin{eqnarray*}
-{\rm Tr}\left[\Phi^{\otimes n}(\rho_{i_n})\log\Phi^{\otimes n}(\tilde{\rho})\right]
+{\rm Tr}\left[\tilde {\Phi}^{\otimes n}(\rho_{i_n})\log\tilde {\Phi}^{\otimes n}(\tilde{\rho})\right]&& \\
-\log  \delta(i_n)  &\leq& I_c (\rho_{i_n}, \Phi^{\otimes n}),
\end{eqnarray*}
from which follows
$$
\lim\frac{1}{n}\max_{\rho_{i_n}} \left\{-{\rm Tr}\left[\Phi^{\otimes n}(\rho_{i_j})\log\Phi^{\otimes n}(\tilde{\rho})\right]+{\rm Tr}\left[   \tilde {\Phi}^{\otimes n}(\rho_{i_n})\log\tilde {\Phi}^{\otimes n}(\tilde{\rho})\right]\right\}  \leq Q_c (\Phi).
$$
\end {proof}

Notice that the quantity inside the brackets in Theorem \ref{th:QonS} (as well as its maximum) can be computed with arbitrary degree of precision. Then as consequence of Theorem ~\ref{th:approx}  finding the recursively approximation of $Q_c$ restricted to the set  of density matrices $\cal S$
results possible.

Furthermore, from an algorithmic point of view, we can borrow from Theorem \ref{th:QonS} a lower bound on the quantum capacity.


\section{Conclusion}\label{sec:conclu}

We have investigated the quantum channel capacity based on the computability concept. In this process,  the von Neumann entropy is replaced by the Gacs entropy \cite{Gacs} which is defined based on the universal semi-measure. The algorithmic coherent information is rewritten in terms of this entropy which results linear with respect to density matrices.

 Then we have shown that quantum channel capacity restricted on semi-computable density matrices is recursively approximable (Theorem \ref{th:approx}). This constitutes a step forward since the negative claim about computability of quantum channel capacity \cite{wolf,cubitt}.
However, since the algorithmic coherent information cannot be computed for any given number of channel's usage due to un-computability of the universal semi-measure, we have subsequently introduced a restriction to compute it with any degree of precision (Theorem \ref{th:QonS}).

We believe it will be useful for computer programmers to know that quantum channel capacity
can be recursively approximated by removing the the maximization of coherent information at each channel's usage level. As well as the fact that restricting to the set of computable density matrices a lower bound on the quantum capacity can be computed.

It remains open the problem of whether exist or not in the plethora of quantum channel capacities one whose algorithmic version differs from the standard one. It seems natural to next address this issue for the classical capacity of quantum channels which requires an analysis of computability of Holevo $\chi$ quantity \cite{holevo}.

Finally, given that the quantum capacity is related to entanglement distillability, it would be worth extending the pursued algorithmic approach to the subject of entanglement manipulation (distillation and dilution). This is left for future investigations.



\end{document}